\documentclass[10pt,twocolumn,twoside]{IEEEtran} 
\usepackage{latexsym,tabularx}
\usepackage{graphicx}
\usepackage{epstopdf}
\graphicspath{{./figure/}}
\DeclareGraphicsExtensions{.pdf,.jpeg,.jpg,.png}
\usepackage{array}
\usepackage{epsfig}
\usepackage{subfig}
\usepackage{psfrag,psfragx}
\usepackage{times}	
\usepackage{tgtermes}	
\usepackage[T1]{fontenc}	
\usepackage{yfonts}
\usepackage{pgothic}
\usepackage{adjustbox}
\usepackage{amsmath,amsfonts,amssymb,array,tabulary,amsthm}
\usepackage{upgreek}
\usepackage{float}
\usepackage{arydshln}
\usepackage{algorithm,algorithmic}
\usepackage{color,url}
\usepackage{framed}
\usepackage{mathdots}
\usepackage{mathrsfs} 
\usepackage{tikz}
\usetikzlibrary{decorations.pathmorphing}
\usetikzlibrary{arrows.meta,automata,positioning,shapes}
\usetikzlibrary{decorations.pathmorphing}
\tikzset{snake it/.style={decorate, decoration=snake}}
\usepackage{pdflscape}
\usepackage{mathtools} \mathtoolsset{showonlyrefs}
\usepackage{theoremref} 
\usepackage{cite}
\newtheorem{theorem}{Theorem}
\newtheorem{proposition}{Proposition}
\newtheorem{lemma}{Lemma}
\newtheorem{corollary}{Corollary}

\newtheorem{definition}{Definition}

\newcommand{\nd}[1][d]{\mathcal{N}_{#1}}
\newcommand{\pd}[1][d]{\mathcal{P}_{#1}}

\newcommand{\G}{{\mathcal{G}}}
\newcommand{\C}{{\mathcal{C}}}
\newcommand{\W}{{W_{c}}}
\newcommand{\AB}{{[A~B]}}

\IEEEoverridecommandlockouts
\title{Herdability of Linear Systems \\ Based on Sign Patterns and Graph Structures}
\author{Sebastian F. Ruf$^{1,+}$, Magnus Egerstedt$^{2}$, and Jeff S. Shamma$^{3}$
	\thanks{$^{1}$S. F. Ruf is with the Center for Complex Networks Research and Department of Psychology, Northeastern University, Boston, MA 02115 and can be reached at
		$858-752-3571$ and {\tt\small ruf@gatech.edu}} %
	\thanks{$+$ Corresponding author.}
		\thanks{$^{2}$ M. Egerstedt is with the School of Electrical and Computer Engineering, Georgia Institute of Technology, Atlanta, GA 30332 and can be reached at $404-894-4468$ and {\tt\small magnus@gatech.edu} }%
	\thanks{$^{3}$ J. S. Shamma is with Computer, Electrical and Mathematical Science and Engineering Division, King Abdullah University of Science and Technology (KAUST), Saudi Arabia, and can be reached at {\tt\small jeff.shamma@kaust.edu.sa}}%
	\thanks{$^{\star}$ This work was partially supported by funding from KAUST}
}

\begin{document}
	\maketitle
	\thispagestyle{empty}
	\pagestyle{empty}
	\begin{abstract}
		We consider the notion of herdability, a set-based reachability condition, which asks whether the state of a system can be controlled to be element-wise larger than a non-negative threshold. First a number of {foundational} results on herdability of a continuous time, linear time invariant system are presented. These show that the herdability of a linear system can be determined based on certain matrices, such as the controllability matrix, which arise in the study of controllability of linear systems. {Second,} the relationship between the sign pattern of the underlying graph structure of a system and the herdability properties of the system is {investigated.} In doing so the notion of sign herdability is introduced which captures classes of systems whose sign pattern determines their herdability. We identify a set of conditions, first on the sign pattern of the controllability matrix and then on the underlying graph structure, {that} ensure that the system is sign herdable. 
	\end{abstract}
\section{Introduction}\label{sec:intro}
Controllability is a fundamental property of a dynamical system; its study began with the work of Kalman et. al in the 1960s \cite{kalmancont} and recently has seen a renewed interest in the study of networked systems \cite{mesbahi2010graph} and complex networks \cite{Liu2011}. However there are cases where a system need not be completely controllable to achieve desirable system outcomes: for example if the system is stabilizable \cite{hespanha2009linear}, if a desired end point can be reached \cite{jadbabaie2019minimal} or equivalently if it is transittable \cite{Wu2014}, or if the system is target {(output)} controllable \cite{Gao2014a}.

 This paper discusses the behavior of a class of systems in the case where complete controllability is not required by considering the reachability of a specific set in the state space, rather than the whole state space as is desired in complete controllability. This approach builds on a history of understanding dynamical systems through the geometry of reachable spaces \cite{wonham,egline} and is also distinct from the approach when considering, for example, a predictive control problem \cite{borrelli2017predictive}; here the goal is {to }understand whether the system can reach a specific set instead of understanding the reachable set of the system.

To see why one might take into account the reachability of a set, consider the case of a dynamical system where the state represents the percentage of a given community that will vote for a political candidate and {in which a political campaign can apply advertising effort as a control signal.} Here {the} political campaign is successful if the state can be driven high enough for the candidate to win, regardless of whether communities can be made to vote at any specific percentage as would be required by complete controllability. 

In order to study systems that are not completely controllable but for which certain desirable control outcomes are still achievable, this paper considers the set-based reachability condition of herdability, which captures whether the components of the state can be driven above a non-negative threshold. {This notion is called herdability to capture cases like herding where the goal is not that each sheep be placed perfectly in the meadow, but that the herd can make it into the pen.} This target set describes desired behavior in social and biological sciences where many systems act based on thresholds, for example collective social behavior \cite{granovetter1978threshold,schelling1971dynamic}, the diffusion of innovations \cite{valente1995network}, quorum sensing \cite{miller2001quorum}, and the firing of a neuron \cite{hodgkin1952quantitative}.

More formally, a continuous time, linear time invariant (LTI) system, referred to from now on simply as a linear system, \begin{equation}\label{eq:sys}
\dot{\mathbf{x}}=A\mathbf{x}+B\mathbf{u}
\end{equation}
where $A\in\mathbb{R}^{n\times n}$ and $B\in\mathbb{R}^{n\times m}$, is \emph{completely herdable} if {for any initial condition} there exists a control input that makes the state enter the set $\mathcal{H}_d = \{ \mathbf{x} \in \mathbb R^{n} : \mathbf{x}_{i} \geq d\}$ for all $d\geq0$, where $ \mathbf{x}_{i}$ is the $i$-th element of $\mathbf{x}$. 
Returning to the example of voting in an election, where $\mathbf{x}_{i}$ now represents the percentage of community $i$ that will vote for a candidate, one can see that a sufficient condition to win the election is to reach the set $\mathcal{H}_{.5} = \{ \mathbf{x} \in \mathbb R^{n} : (\mathbf{x})_{i} \geq .5\}$. 

This paper begins by introducing theoretical results on the herdability of a {linear} system, which are related to standard results on the controllability of LTI systems \cite{hespanha2009linear}. After the introduction of these results, this paper considers the herdability of LTI systems based on the structure of the underlying interaction graph, which encodes information about how states and inputs interact with each other based on the system matrices $A$ and $B$. The relationship between a graph and a dynamical system which evolves over that graph has previously received considerable attention, see for example \cite{mesbahi2010graph,liu2015control}. The existing work can be roughly classified into two primary approaches. 

The first approach translates a specific graph structure to system dynamics, which in many cases are consensus dynamics. Consensus dynamics are used in the context of multi agent robotic and social systems \cite{mesbahi2010graph,abelson1964mathematical}. The controllability of these consensus systems has been shown to be directly related to the structure of the underlying graph \cite{ rahmani2009controllability,martini2010controllability,nabi2013controllability,parlangeli2012reachability,notarstefano2013controllability}. System controllability is lost when nodes are symmetric with respect to an input, where symmetry is discussed here in terms of the automorphism structure of the graph \cite{ rahmani2009controllability,martini2010controllability,chapman2014symmetry}. Symmetry appears in the context of structural controllability as a dilation\cite{Lin1974}. In the case of herdability, it is possible to have certain types of symmetry without degrading the herdability of a system.  

 The second approach takes classes of linear dynamical systems and maps them to a graph to discuss properties of all systems that share the same graph structure. This approach is known generally as qualitative systems analysis, {and was introduced to deal with the inherent uncertainty in system parameters. Qualitative systems analysis} considers structural controllability, \cite{Lin1974,shields1976structural,glover1976characterization}, its extension strong structural controllability \cite{mayeda1979strong}, and sign controllability \cite{johnson1993sign,hartung2013characterization,tsatsomeros1998sign,olesky1993qualitative}. {Structural controllability and strong structural controllability were introduced to describes systems where an interaction between state variables is known to be present but of unknown magnitude. } In structural controllability, a dynamical system is represented by a graph in which each edge of the graph is assigned a weight in $\mathbb{R}$. A system is structurally controllable (strongly structurally controllable) if and only if it is controllable for almost all (all) weights that are assigned to the edges, which is a property that can be verified directly from the structure of the underlying graph. Structural controllability has seen application in the study of complex networks, see \cite{liu2015control} for an extended discussion.
  
 {Sign controllability takes into account more information than structural controllability, assuming that both the presence and the sign of an interaction between state variables is known. In a linear system, this translates into using the sign pattern of the system matrices to determines controllability.} {The application of sign pattern to understand controllability} builds on two sets of results from the economics and ecology literature. The first is sign stability \cite{quirk1965qualitative,may1973qualitative}, which asks whether a matrix is stable based on its sign pattern. {Sign stability has been applied to understand biological systems, for example \cite{blanchini2014structural}.} The other is sign solvability \cite{brualdi2009matrices}; which asks, when solving the matrix equation $A\mathbf{x}=\mathbf{b}$, whether the sign pattern of $\mathbf{x}$ is uniquely determined by the sign pattern of $A$ and $\mathbf{b}$.
 
 Sign controllability has been considered with regard to the structure of the various system matrices \cite{johnson1993sign, hartung2013characterization} as well as the sign pattern of the underlying graph \cite{tsatsomeros1998sign}. {To date, the results on sign controllability have been limited. Results for the case of a non-negative $A$ matrix are considered in \cite{johnson1993sign}, while \cite{hartung2013characterization} considers $A$ matrices that have real eigenvalues. {Reference} \cite{tsatsomeros1998sign} extends to conditions based on the underlying graph structure, however the results depend on a decision problem which was shown to be NP-complete in \cite{klee1984signsolvability}.}

{Although there has not been success in applying the study of sign controllability to complex networks, there is much recent interest to understand system behavior based on sign pattern.} Signed graphs are used in the social networks context to represent systems in which agents are both friends and enemies \cite{wasserman1994social,easley2010networks}. These networks have also been of recent interest in the controls community \cite{altafini2013consensus,alemzadeh2017controllability}. 

{In this work, we consider the problem of determining system herdability based on sign pattern for two reasons. First, the applications of interest for herdability are those that are inherently uncertain but where sign pattern can be reasonably known, i.e. biological and social systems. Second, the results of first half of the paper show that sign pattern can be used to characterize the herdability of a system. 
	
	In that light, }the second half of this paper shares the approach of sign controllability in that the control properties of classes of systems are considered based first on the sign pattern of the controllability matrix and then on the sign pattern of the underlying graph structure. Specifically this paper represents the interaction structure of a linear system as a signed, directed graph and explores when the sign pattern is sufficient to determine system herdabiliy. In this light, the central problem of the paper can be phrased in a social networks context as follows: how does the grouping of friends and enemies in the network relate to the ability to convince agents in the system to hold a high opinion?

	 		{This paper presents theoretical underpinnings of the notion of herdability and extends the conference paper \cite{rufacc} as follows:}
	 		\begin{itemize}
	 			\item {Basic theoretical results on the notion of herdability are developed.} These results take the form of conditions on the range of the controllability matrix $\C$ {as well as conditions on the controllability grammian $\W$, and the matrix $\AB$, which were not included in the conference paper}. These results {are extended to novel} sufficient conditions on the columns of the controllability matrix $\C$.
	 			
	 			\item{Conditions which ensure that the controllability matrix is sign-definite are introduced.}
	 			\item The basic characterization of herdability {is extended } to conditions based on the underlying graph of the LTI system. These include a necessary condition for herdability, extensions of the {novel }sufficient conditions for complete herdability, and a complete characterization of the herdability properties of systems which have an underlying graph which is a directed out branching. 
	 		\end{itemize}
The rest of the paper is organized as follows: Section~\ref{sec:herd} introduces the basic theory of herdable system. This includes the necessary definitions and necessary and sufficient conditions for herdability. Sufficient conditions for herdability are considered in Section \ref{sec:class}. In Section~\ref{sec:pre} a graph theoretic characterization of the interaction structure of a linear system is presented. Section \ref{sec:graph_con} considers graphical conditions for herdability including a neccesary condition in Section \ref{sec:nes}, sufficient conditions in Section \ref{sec:suf_graph} and selecting a herdable subset for graphs that are represented by a directed out-branching in Section \ref{sec:outb}. The paper concludes in Section~\ref{sec:conc}.
	 
	 	 \subsection{Notation:}\label{sec:note} For a vector $\mathbf{k}\in\mathbb{R}^{n}$, $\mathbf{k}_{i}$ refers to the $i$-th element of $\mathbf{k}$. For a matrix $K\in\mathbb{R}^{n\times m}$, $K_{i,:}$ refers to the $i$-th row of $K$, $K_{:,j}$ refers to the $j$-th column of $K$ and $K_{i,j}$ to the $i,j$-th element of $K$. The cardinality of the set $\mathcal{S}$ is expressed as $|S|$. Let $\mathrm{sgn}(\cdot)$ denote the sign function which is defined as 
	 	 \begin{equation}
	 	 \mathrm{sgn}(x) = \left\{\begin{array}{lr}
	 	 -1 & \text{for } x<0,\\
	 	 0 & \text{for } x=0,\\
	 	 1 & \text{for } x>0.
	 	 \end{array}\right.
	 	 \end{equation}  Let $\mathbf{0}_{n}\in\mathbb{R}^{n}$ be a vector of zeros, $\mathbf{1}_{n}\in\mathbb{R}^{n}$ be a vector of ones, and $0_{n\times m}\in\mathbb{R}^{n\times m}$ be a matrix of zeros. Logical AND is denoted by $\wedge$ and $\vee$ denotes logical OR and $\veebar$ denotes logical EXCLUSIVE OR. A vector is balanced if it is the zero vector or contains both positive and negative elements. A vector is unisigned if its non-zero elements all have the same sign. A unisigned vector is positive (negative) if all non-zero elements have a positive (negative) sign. 
	 	\section{Herdability {Based On Sign Patterns}}\label{sec:herd}
	 	In this section, the basic theory of the herdability of continuous time, linear dynamical systems is presented as well as a characterization of herdability based on matrices which relate to the reachability properties of a linear system. We begin with the formal definition of herdability:  
	 	\begin{definition}
	 		The state $x_i$ of a linear system is $d$-herdable if $\forall \mathbf{x}(0)\in\mathbb{R}^{n}$ , there exists a finite time $t_{f}$ and an input $\mathbf{u}(t), \ t\in[0,t_{f}]$ such that $\mathbf{x}_{i}(t_{f})\geq d$ under control input $\mathbf{u}(t)$.
	 	\end{definition}
	 	
	 	If the state is $d$-herdable for any $d\geq 0$ it will be said to be herdable. In the case of linear systems, $d$-herdability for $d>0$ and herdability are equivalent. This is because if there exists an input $u^{\ast}$ which can drive the system to be larger than some $d^{\ast}>0$ then that input can be scaled by some positive constant so that the resulting state is element wise larger than any $d>0$. As the following discussion concerns itself with the analysis of linear systems, we will refer only to the herdability of states and where appropriate the herdability of the complete linear system.
	 	
	 	\begin{definition}
	 		The state $x_i$ of a linear system is herdable if $\forall \mathbf{x}(0)\in\mathbb{R}^{n}, h\geq 0$, there exists a finite time $t_{f}$ and an input $\mathbf{u}(t), \ t\in[0,t_{f}]$ such that $\mathbf{x}_{i}(t_{f})\geq h$ under control input $\mathbf{u}(t)$.
	 	\end{definition}
	 	\begin{definition}
	 		A {sub}set of states, $\mathcal{X}\subseteq\{x_1,x_2,\dots,x_n\}$, is herdable if each individual state in $\mathcal{X}$ is herdable together, i.e. if $\forall \mathbf{x}(0)\in\mathbb{R}^{n}$ and $h\geq 0$, there exists a finite time $t_{f}$ and an input $\mathbf{u}(t), \ t\in[0,t_{f}]$ such that $\mathbf{x}_{i}(t_{f})\geq h, ~\forall x_i\in \mathcal{X}$ under control input $\mathbf{u}(t)$. {A linear system is completely herdable if all states in the system are herdable together.} 
	 	\end{definition} 
	 	
{For small systems with known dynamics, one could directly characterize the reachable subspace. For the potential application domains for herdability such characterizations become difficult. Instead of characterizing the whole reachable subspace, we consider whether there is any part of it that intersects the positive orthant.   } {To do so will} require some basic concepts from the study of linear systems, specifically the relation between the reachable subspace and the controllability grammian $\W$ and controllability matrix $\C$. This section follows \cite{hespanha2009linear} but any text on linear systems theory will do.  
	 	
	 	\begin{definition}
	 		The reachable subspace $\mathcal{R}[0,t]$ of a linear system is
	 			 	 \begin{equation}
	 			 	 \begin{aligned}
	 		\mathcal{R}&[0,t]= \\ &\left\{\mathbf{x}_{1}\in\mathbb{R}^{n}:\exists \mathbf{u}:[0,t]\rightarrow\mathbb{R}^{m},~ \mathbf{x}_{1}=\int_{0}^{t}e^{A(t-\tau)}B\mathbf{u}(\tau)d\tau\right\}.
	 		\end{aligned}
	 			 		\end{equation}  
	 		The controllability matrix $\mathcal{C}$ of a linear system is 
	 		\begin{equation}
	 		\mathcal{C}=\left[B,AB,A^{2}B,\dots,A^{n-1}B \right].
	 		\end{equation}

	 		The Controllability Grammian on the time interval $[0,t]$, $W_{c}[0,t]$, of a linear system is \begin{equation}W_c[0,t]=\int_{0}^{t}e^{A\tau}BB^{T}e^{A^{T}\tau}d\tau.\end{equation} \end{definition}

	 	\begin{lemma}(Theorem 11.5 \cite{hespanha2009linear}) \thlabel{lem:rceq} \begin{equation}
	 		\mathcal{R}[0,t]=\mathrm{range}(\mathcal{C})=\mathrm{range}(W_{c}[0,t]).
	 		\end{equation} 
	 	\end{lemma}     

	 	Note that in linear systems, the reachable subspace does not depend explicitly on the time interval used and as such the time interval will be omitted for notational convenience. 
	 	
	 	{Similar to the classical conditions for controllability, it is possible to use $\C$ and $\W$ to develop conditions for herdability.}
	 	
	 	\begin{theorem}\thlabel{lem:herdset}
	 		A set of states $\mathcal{X}\subseteq\{x_1,x_2,\dots,x_n\}$ in a linear system is herdable if and only if there is exists a vector $\mathbf{k}\in \mathrm{range}(\mathcal{C})$ that satisfies $\mathbf{k}_{i}>0$ for all $x_i\in\mathcal{X}$.
	 	\end{theorem}
	 	\begin{proof}
	 		Define the set  $\mathcal{K}$ to be the set that contains the positive elements of $\mathbf{k}$, $\mathcal{K}=\{p \ | \ p>0 \ \wedge \ \exists \ x_i \text{ such that } \mathbf{k}_{i}=p\}.$ 
	 		
	 		\noindent($\mathbf{k}\in \mathrm{range}(\mathcal{C})\Rightarrow$ $\mathcal{X}$ is herdable) Consider the problem of controlling all states in the set $\mathcal{X}$ to be greater than some lower threshold $h\geq 0$ from an initial condition $\mathbf{x}(0)$. Suppose there is a $\mathbf{k}\in \mathrm{range}(\mathcal{C})$, that satisfies $\mathbf{k}_{i}>0$ if $x_i\in\mathcal{X}$. As $\mathbf{k}\in \mathrm{range}(\mathcal{C})$, $\exists \pmb{\alpha}$ such that \begin{equation}
	 		\mathcal{C}\pmb{\alpha}=\mathbf{k}.
	 		\end{equation}  If \begin{equation}\gamma>\frac{\max_{j} \  (h\mathbf{1}_{n}-e^{At}\mathbf{x}(0))_{j}}{\min{\mathcal{K}}}\end{equation} and $\mathbf{v}=\gamma\pmb{\alpha}$ then for all $x_i\in \mathcal{X}$ it holds that
	 		\begin{equation}
	 		(\mathcal{C}\mathbf{v})_{i} > (h\mathbf{1}_{n}-e^{At}\mathbf{x}(0))_{i}.
	 		\end{equation} 
	 		As the range of $\mathcal{C}$ is the same as the reachable subspace by \thref{lem:rceq}, $\exists \mathbf{u}(\cdot)$ such that for all $x_i\in \mathcal{X}$ 
	 		\begin{equation}(e^{At}\mathbf{x}(0)+\int_{0}^{t}e^{A(t-\tau)}B\mathbf{u}(\tau)d\tau)_{i}>h
	 		\end{equation}
	 		then all states in $\mathcal{X}$ can be made larger that $h$ and as $h$ is arbitrary the subset of states $\mathcal{X}$ is herdable.
	 		
	 		\noindent($\mathcal{X}$ is herdable $\Rightarrow\mathbf{k}\in \mathrm{range}(\mathcal{C})$) 
	 		As the set of state nodes $\mathcal{X}$ is herdable, each element of $\mathcal{X}$ can be made larger than some $h^{\ast}>0$ from any initial condition. Consider the initial condition $x(0)=\mathbf{0}_{n}$. Then by the herdability of the set $\mathcal{X}$ there exists a vector $\mathbf{k}^{\ast}$ that satisfies $\mathbf{k}^{\ast}_{i}>h^{\ast} \ ~\forall x_i\in \mathcal{X}$ and an input $\mathbf{u}(\cdot)$ such that \begin{equation}\int_{0}^{t}e^{A(t-\tau)}B\mathbf{u}(\tau)d\tau=\mathbf{k}^{\ast}
	 		\end{equation} Then $\mathbf{k}^{\ast}_{i}>0 \ ~\forall x_i\in\mathcal{X}$ by the definition of $h^{\ast}$. By the definition of $\mathcal{R}[0,t]$, $\mathbf{k}^{\ast}\in \mathcal{R}[0,t]$ and consequently $\mathbf{k}^{\ast}\in \mathrm{range}(\mathcal{C})$ by \thref{lem:rceq}.
	 	\end{proof}
	 	
	 	\begin{corollary}\thlabel{cor:krange}
	 		A linear system is completely herdable if and only if there exists an element-wise positive vector $\mathbf{k}\in \mathrm{range}(\mathcal{C})$. 
	 	\end{corollary}
	 	A similar statement can be made about the controllability grammian $W_{c}$ of a system, following directly from \thref{lem:rceq} and \thref{lem:herdset}. 
	 	\begin{corollary}\thlabel{lem:herdsetwc}
	 		A set of states $\mathcal{X}\subseteq\{x_1,x_2,\dots,x_n\}$ in a linear system is herdable if and only if there is exists a vector $\mathbf{k}\in \mathrm{range}(W_{c})$ that satisfies $\mathbf{k}_{i}>0$ for all $x_i\in\mathcal{X}$. A linear system is completely herdable if and only if there exists an element-wise positive vector $\mathbf{k}\in\mathrm{range}(W_{c})$.
	 	\end{corollary}
	 	There is also a necessary condition for herdability which arises based on the characterization of \thref{lem:herdset}. 
	 	\begin{theorem}\thlabel{th:abn}
	 		If a linear system is completely herdable then there exists an element-wise positive vector $\mathbf{k}\in\mathrm{range}([A \ B])$.
	 	\end{theorem}
	 	\begin{proof}
	 		If a linear system is completely herdable, then by \thref{lem:herdset}, there is an element-wise positive vector $\mathbf{k}\in\mathrm{range}(\mathcal{C}).$ As such there exists a $\mathbf{y}\in\mathbb{R}^{nm}$ such that $$\mathcal{C}\mathbf{y}=\mathbf{k}.$$ Dividing $\mathbf{y}$ into $n$ subcomponents, with each $\mathbf{y}_{i}\in\mathbb{R}^{m}$: \begin{equation}
	 		\mathbf{y}=\begin{bmatrix}
	 		\mathbf{y}_{1}\\
	 		\mathbf{y}_{2}\\
	 		\vdots \\
	 		\mathbf{y}_{n}
	 		\end{bmatrix}
	 		\end{equation}
	 		gives that \begin{equation}
	 		\begin{aligned}
	 		\mathbf{k}&=\mathcal{C}\mathbf{y}\\
	 		&=B\mathbf{y}_{1}+AB\mathbf{y}_{2}+\dots+A^{n-1}B\mathbf{y}_{n}\\
	 		&=B\mathbf{y}_{1}+A(B\mathbf{y}_{2}+AB\mathbf{y}_{3}+\dots+A^{n-2}\mathbf{y}_{n}).\\
	 		\end{aligned}
	 		\end{equation}
	 		Then $\mathbf{k}\in\mathrm{range}([A~B])$ as
	 		\begin{equation}
	 		\mathbf{k}=\begin{bmatrix}
	 		A & B
	 		\end{bmatrix} \begin{bmatrix}
	 		B\mathbf{y}_{2}+AB\mathbf{y}_{3}+\dots+A^{n-2}\mathbf{y}_{n} \\ \mathbf{y}_{1}
	 		\end{bmatrix}.
	 		\end{equation}
	 	\end{proof}
	 	
	 	{Both \thref{lem:herdsetwc} and \thref{th:abn} can provide important information when considering large-scale systems with known dynamics. The infinite horizon controllability grammian can be calculated in linear time if the $A$ matrix is stable, hence \thref{lem:herdsetwc} provides a tractable method to determine herdability computationally.} \thref{th:abn} can be used as an initial check of system herdability when dealing with large systems as there is virtually no computational overhead to generate the matrix $\AB$. {As this paper seeks to understand herdability based on the underlying graph structure, we will primarily discuss the results based on the controllability matrix; which will be shown later to have a natural interpretation based on the underlying graph.}
	 	
	 	{While this paper seeks to understand whether a system is herdable under a fixed $B$, it's interesting to note that \thref{th:abn} could} be used for designing the interaction with the system via the selection of a $B$ matrix. In the case that the $A$ matrix is such that there is no element-wise positive $\mathbf{k}\in\mathrm{range}(A)$ then $B$ can be designed such that there is an element-wise positive  $\mathbf{k}\in\mathrm{range}(\AB)$. However as \thref{th:abn} is a necessary condition, other more computationally expensive methods would be required to verify that the system is actually herdable. 
	 	
	 	{If the system dynamics are known,} determining whether there is some positive unisigned vector $\mathbf{k}$ in the range of a given matrix of interest, be it $\mathcal{C}$, $W_{c}$ or $[A~B]$, can be done simply by finding a basis for the matrix. However as the scale of the system increases and the individual system parameters become uncertain, determining the element-wise positive vector $\mathbf{k}$ which shows that the system is herdable is non-trivial. Many of the results that will follow provide sufficient conditions for existence of a positive unisigned $\mathbf{k}$. 
\subsection{{Relation to Sign Solvability}}
{Before preceding with the characterization of herdability, it's worth introducing concepts from the study on sign solvability \cite{brualdi2009matrices}. The matrix equation $Ax=b$ is sign solvable if for all matrices $A$ and $b$ with the same sign pattern, a solution $x$ exists and all solutions have the same sign pattern. The matrix equation $Ax=b$ is strongly sign solvable if it is solvable and $x$ must be elementwise nonzero. Finally, the matrix equation $Ax=b$ is sign inconsistent if for all $A$ and $b$ with the same sign pattern no solution $x$ exists such that $Ax=b$ \cite{shao1999sign}. These definitions can be translated to conditions for herdability.
	\begin{proposition}
		If $\mathcal{C}x=1_{n}$ is sign inconsistent then $(A,B)$ is not completely herdable. If $\mathcal{C}x=1_{n}$ is sign solvable or strongly sign solvable then $(A,B)$ is completely herdable. 
	\end{proposition}
	Verifying these conditions involves studying L-matrices and S-matrices. An L-matrix has linearly independent rows based on sign pattern and an S-matrix, $\hat{S}$ has an elementwise positive vector $k\in\mathrm{null}(\hat{S})$. If $[\mathcal{C} ~ 1_{n}]$ is an L-matrix then $\mathcal{C}x=1_{n}$ is sign inconsistent \cite{shao1999sign}. Unfortunately verifying that a matrix is an L-matrix is NP-complete \cite{klee1984signsolvability}. 
	
	If $[\mathcal{C}~ -1_{n}]$ is an S-matrix then $\mathcal{C}x=1_{n}$ is strongly sign solvable. The theory for S-matrices was developed for matrices $\hat{S}\in\mathbb{R}^{n\times n+1}$ which implies that these methods can only be used to check herdability in the case that $m=1$ and $\C\in\mathbb{R}^{n\times n}$. In the single input case, strong sign solvability can be verified in polynomial time \cite{klee1987recursive}. 
	
	While these conditions based on sign solvability are sufficient for complete herdability, they are often overly restrictive. The conditions presented in the remainder of this paper for herdability consider whether a solution exists to $\mathcal{C}x=k,$ for some elementwise positive $k$ without the restriction that the sign pattern of such solutions be uniquely determined. It's also worth noting that the approaches of sign solvability often implicitly involve verifying that $\C$ is sign non singular, i.e. that the system is sign controllable.   }
	 	\subsection{Sufficient Conditions for Herdability}\label{sec:class}
	 	The section provides a number of sufficient conditions for herdability based on the structure of the controllability matrix $\mathcal{C}$ and controllability grammian $W_{c}$. 
	 	
	 	To do so requires the following set of definitions from the study of qualitative systems {\cite{brualdi2009matrices}}, which we recall from Section \ref{sec:note}. A vector is balanced if it is the zero vector or contains both positive and negative elements. A vector is unisigned if its non-zero elements all have the same sign. A unisigned vector is positive (negative) if all non-zero elements have a positive (negative) sign. 
	 	
	 	\begin{definition}A state $x_i$ in the system is \textbf{strictly herdable}, if $\exists \mathbf{k} \in \mathcal{R}[0,t]$ such that $\mathbf{k}$ is unisigned and $\mathbf{k}_{i}\neq0$.
	 		A state $x_i$ is \textbf{loosely herdable} if all vectors $\mathbf{k} \in \mathcal{R}[0,t]$ such that $\mathbf{k}_{i}\neq0$ are balanced.
	 	\end{definition}  
 	 
 		{The presence of a loosely herdable node can be verified by showing that $\C x=1_{n}$ is sign inconsistent, which as mentioned previously is NP-Complete.} As such this section focuses on verifying that a state is strictly herdable based on the sign pattern of the controllability matrix $\C$.  
	 	\begin{lemma}\thlabel{lem:s}
	 		{Let $\mathcal{S}\subseteq\{x_{1},x_{2},\dots,x_{n}\}$ be the subset of states such that for all $x_{i}\in\mathcal{S}$ there exists a unisigned column of $\mathcal{C}$ with a non-zero element at position $i$.} Each element of $\mathcal{S}$ is strictly herdable.
	 	\end{lemma} 
	 	\begin{proof}
	 		By the definition of $\mathcal{S}$, for node $x_{is}\in\mathcal{S}$ there exists a $j_{s}$ such that $(\mathcal{C})_{i_{s},j_{s}}\neq0$ and each non-zero element of $(\mathcal{C})_{:,j_{s}}$, has the same sign. If $(\mathcal{C})_{i_{s},j_{s}}>0$, then  $(\mathcal{C})_{:,j_{s}}\in\mathrm{range}(\mathcal{C})$ and the node $x_{is}$ is strictly herdable.  Alternatively if $(\mathcal{C})_{i_{s},j_{s}}<0$, then the positive unisigned vector $-(\mathcal{C})_{:,j_{s}}\in\mathrm{range}(\mathcal{C})$ and the node $x_{is}$ is strictly herdable.
	 	
	 	\end{proof}
	 	Let $\mathcal{D}=\{x_{1},x_{2},\dots,x_{n}\}\setminus\mathcal{S}$. If $x_l\in\mathcal{D}$ then for all $j$ such that $(\mathcal{C})_{l,j}\neq 0$ the column vector $(\mathcal{C})_{:,j}$ is balanced. This holds because if this were not true then $x_{l}$ would be in $S$. We introduce the following definition to classify states in $\mathcal{D}$.  \begin{definition} The state $x_z$ \textbf{balances} state $x_l$ at $j$ if it has a different sign than $x_l$ in the column $(\mathcal{C})_{:,j}$ and \textbf{favors} state $x_l$ at $j$ if it has the same sign as $x_l$ in the column $(\mathcal{C})_{:,j}$. \end{definition}
	 	\begin{lemma}\thlabel{lem:shat}
	 		If $\forall x_l\in\mathcal{D}$ there exists a $j$ such that $x_l$ is balanced only by strictly herdable nodes at $j$ then $x_l$ is strictly herdable. 
	 	\end{lemma} 
	 	\begin{proof}
	 		Let $\hat{\mathcal{S}}$ be the set of nodes which balance $x_l$ at $j$, which by assumption are all strictly herdable. By the definition of strictly herdable nodes, for each $x_s\in\hat{\mathcal{S}}$ there exists a vector $\mathbf{v}^{s}\in\mathrm{range}(\mathcal{C})$ such that $\mathbf{v}^{s}_{s}> 0$ and each non-zero element of the vector $\mathbf{v}^{s}$ has the same sign. Consider the set of vectors $S=\{\{\mathbf{v}_{s}\}_{x_s\in\hat{\mathcal{S}}},\mathbf{b}\}$ where $\mathbf{b}=\mathcal{C}_{:,j}$, the vector where $x_l$ is opposed by the elements of $\hat{\mathcal{S}}$. Then there exists a collection of weights $\alpha_{s}$ such that $\hat{\mathbf{s}}=\sum_{x_s\in\hat{\mathcal{S}}} \alpha_{s}\mathbf{v}^{s} + \mathrm{sgn}(b_{l})\mathbf{b}$ is a vector which is positive at $x_l$, at each node that favors $x_l$ at $j$ and at each node $x_s\in\hat{\mathcal{S}}$. As $\hat{\mathbf{s}}\in\mathrm{range}(\mathcal{C}),$ $l$ is strictly herdable. 
	 	\end{proof}
	 	Theorem \ref{th:suff} below lays the foundation for using Lemma \ref{lem:shat} to discuss system herdability.
	 	\begin{theorem}\thlabel{th:suff} 
	 		All states $x_i \in \{x_1,x_2,\dots,x_n\}$ are strictly herdable if and only if the system is completely herdable. 
	 	\end{theorem}
	 	\begin{proof}
	 		(Sufficiency) As each state $x_i\in \{x_1,x_2,\dots,x_n\}$ is strictly herdable, there exists a vector $\mathbf{k}^{i}\in\mathrm{range}(\mathcal{C})$ which is element-wise non-negative and $\mathbf{k}^{i}_{i}>0$. Then the element wise positive $\mathbf{k}=\sum_{i} k^{i} \in\mathrm{range}(\mathcal{C})$ and the system is completely herdable. 
	 		
	 		\noindent(Necessity) As the system is completely herdable, there is an element-wise positive vector $\mathbf{k}\in\mathrm{range}(\mathcal{C})$. Then for each state $x_i \in \{x_1,x_2,\dots,x_n\}$ $\mathbf{k}_{i}>0$ and the other elements are nonnegative, so state $x_i$ is strictly herdable. 
	 	\end{proof}
 	These results provide a way to check for the herdability of a system efficiently from the controllability matrix, simply by inspecting the columns of $\mathcal{C}$. A similar set of results hold for the columns of the controllability grammian $W_{c}$ though they are not described here. As will be seen shortly, the controllability matrix has the advantage of being related to the underlying graph structure of the network, which presents further opportunities for determining system herdability. 

	 	\section{Characterizing Dynamical Systems via Graphs}\label{sec:pre}
	 	This section presents a characterization of a dynamical system based on the signed, directed graph which describes the interaction structure of the dynamical system. This characterization will allow an exploration of the relationship between the ability to control a system and the structure of the interactions between the states as well as the interaction between the inputs and the states of the system. {In this section, it is shown that a property known as structural balance is sufficient to ensure that all systems with the same sign structure have a controllability matrix with the same sign pattern.}
	 	
	 	A linear system can be represented by three graphs, each of which contains different levels of information about the interactions between the states and inputs. The first is an unweighted, unsigned directed graph $\G=(\mathcal{V},\mathcal{E})$, where $\mathcal{V}$ is the vertex (equivalently node) set and $\mathcal{E}$ is the edge set. This graph is commonly used in the study of structural controllability to represent a class of systems which share the same structure. The second graph is a signed graph $\G^{s}=(\mathcal{V},\mathcal{E},s(\cdot))$ where $s(\cdot)$ accepts an edge and returns a label in $\{+1,-1\}$, which is the sign of the edge. This signed graph represents a class of systems whose edge weights have the same sign pattern. Similarly this representation was used in the study of sign controllability to represent a class of systems which share the same sign structure.  This graph is used when considering structural balance \cite{harary1953notion,harary2005structural}. 
	 	The third graph is a weighted graph $\G^{w}=(\mathcal{V},\mathcal{E},w(\cdot))$ where $w(\cdot)$ accepts an edge and returns a weight in $\mathbb{R}$. The weighted graph is the representation of a single system. 
	 	
	 	As will be seen later, the weighted graph $\G^{w}$ can be directly related to the controllability matrix $\mathcal{C}$ and therefore the controllability properties of the system. The following sections focus on the interplay between $\G^{s}$ and $\G^{w}$, in that the presented structural results are cases where the results for the herdability of a system based on the weighted $\G^{w}$ can be extended to all signed graphs with the same sign structure $\G^{s}$ regardless of the weights of the edges in $\G^{w}$, a notion similar to strong structural controllability and sign controllability. This notion is called sign herdability. 
	 	
	 	\begin{definition}
	 		A system is completely sign herdable if all systems which share the same sign structure $\G^{s}$ are completely herdable.	
	 	\end{definition}
	 {In order to show sign herdability, we will consider when all systems which share the same sign structure $\G^{s}$ give rise to controllability matrix $\C$ or grammian $\W$ which have the same sign pattern.   \begin{definition}
	 		A matrix is sign definite if all systems which share the same sign structure $\G^{s}$ give rise to the same sign structure for that matrix.	
	 \end{definition}}
	 	The formal definition of the graphs follows. The set of vertices satisfies $\mathcal{V}=\mathcal{V}_{x}\cup\mathcal{V}_{u}, \ \mathcal{V}_{x}\cap\mathcal{V}_{u}=\emptyset,$ where $\mathcal{V}_{x}=\{v_{x1},v_{x2},\dots,v_{xn}\}$ is a set of vertices representing the states of the system and $\mathcal{V}_{u}=\{v_{u1},v_{u2},\dots,v_{um}\}$ is a set of nodes representing the inputs to the system. An arbitrary element of $\mathcal{V}$ will be referred to by $v_{i}$ for some index $i$, as will arbitrary elements $v_{xi}\in\mathcal{V}_{x}$ and $v_{ui}\in \mathcal{V}_{u}$. The state $x_i$ will now be interchangeably referred to by the node $v_{xi}$ as will the input $j$ and the node $v_{uj}$.  
	 	
	 	The edge set satisfies $\mathcal{E}=\mathcal{E}_{x}\cup\mathcal{E}_{u}$ where the edges in $\mathcal{E}_{x}$ represent interactions between states of the system, while the edges in $\mathcal{E}_{u}$ represent interactions between the inputs and the states. Denote the directed edge from $v_{i}$ to $v_{j}$ as $(v_{i},v_{j})$. Then $(v_{xi},v_{xj})\in\mathcal{E}_{x} \Leftrightarrow A(j,i)\neq 0$ and $(v_{ui},v_{xj})\in\mathcal{E}_{u} \Leftrightarrow B(j,i)\neq 0$. An arbitrary element of $\mathcal{E}$ will be referred to by $e_{i}$ for some index $i$. By partitioning the node and edges sets, it is possible to define the state subgraph $\G_{x}=(\mathcal{V}_{x},\mathcal{E}_{x})$, which captures only interactions between states as well as the input subgraph $\G_{u}=(\mathcal{V},\mathcal{E}_{u})$ which captures interactions from the inputs to the states. Note that the input nodes do not interact with each other nor is it possible in the present discussion to have an edge of the form $(v_{xi},v_{uj})$, which would imply that the states influences the evolution of the input. 
	 	
	 	When considering the signed graph $\G^{s}$,  $s((v_{xi},v_{xj})) =\mathrm{sgn}(A(j,i))$ and $s((v_{ui},v_{xj})) =\mathrm{sgn}(B(j,i))$. Similarly for $\G^{w}$,   $w((v_{xi},v_{xj}))=A(j,i)$ and $w((v_{ui},v_{xj}))=B(j,i)$. 
	 	
	 	As an example, consider the system
	 	\begin{equation}\label{eq:ex}
	 	\dot{\mathbf{x}}=\begin{bmatrix}
	 	-1 & 0 & 0\\
	 	5 & 0 & 2 \\
	 	4 & -3&  0
	 	\end{bmatrix}+\begin{bmatrix} 0& -2\\
	 	2& 0\\
	 	0& 3
	 	\end{bmatrix}\mathbf{u}
	 	\end{equation} which is translated into $\G^{s}$ and $\G^{w}$ in Figure~\ref{fig:ex}.
	 	\begin{figure}[h]
	 		\centering

	 		\subfloat[]{\label{subfig:seq2}\begin{tikzpicture}[
	 			->,
	 			shorten >=2pt,
	 			auto,
	 			node distance=0.5cm,
	 			every text node part/.style={align=center},
	 			scale=0.25, every node/.style={scale=0.75}
	 			]
	 			\tikzstyle{every state}=[fill=none,draw=black,text=black]
	 			\node[
	 			state,
	 			] (n0) at(2,-1.5)
	 			{$u_{2}$};
	 			\node[
	 			state,
	 			] (n1) at(-12,-1.5)
	 			{$u_{1}$};
	 			\node[
	 			state,
	 			] (n2) at(0,-12)
	 			{$x_{3}$};
	 			\node[
	 			state,
	 			] (n3) at(-12,-12)
	 			{$x_{2}$};
	 			\node[
	 			state,
	 			] (n4) at(-6,-4.5)
	 			{$x_{1}$};
	 			\path (n0) edge [->] node [left,above] {$-$} (n4)
	 			(n0) edge [->] node [left] {$+$} (n2)
	 			(n1) edge [->] node [left] {$+$} (n3)
	 			(n4) edge [->,loop above] node [left,above] {$-$} (n4)
	 			(n2) edge [->,bend left] node [left,above] {$+$} (n3)
	 			(n3) edge [->,bend left] node [left,above] {$-$} (n2)
	 			(n4) edge [->] node [left,above] {$+$}  (n2)
	 			(n4) edge [->] node [left,above] {$+$}  (n3);
	 			\end{tikzpicture}} 
	 		\hspace{.5cm}
	 		\subfloat[]{\label{subfig:seq3}\begin{tikzpicture}[
	 			->,
	 			shorten >=2pt,
	 			auto,
	 			node distance=0.5cm,
	 			every text node part/.style={align=center},
	 			scale=0.25, every node/.style={scale=0.75}
	 			]
	 			\tikzstyle{every state}=[fill=none,draw=black,text=black]
	 			\node[
	 			state,
	 			] (n0) at(2,-1.5)
	 			{$u_{2}$};
	 			\node[
	 			state,
	 			] (n1) at(-12,-1.5)
	 			{$u_{1}$};
	 			\node[
	 			state,
	 			] (n2) at(0,-12)
	 			{$x_{3}$};
	 			\node[
	 			state,
	 			] (n3) at(-12,-12)
	 			{$x_{2}$};
	 			\node[
	 			state,
	 			] (n4) at(-6,-4.5)
	 			{$x_{1}$};
	 			\path (n0) edge [->] node [left,above] {$-2$} (n4)
	 			(n0) edge [->] node [left] {$3$} (n2)
	 			(n1) edge [->] node [left] {$2$} (n3)
	 			(n4) edge [->,loop above] node [left,above] {$-1$} (n4)
	 			(n2) edge [->,bend left] node [left,above] {$2$} (n3)
	 			(n3) edge [->,bend left] node [left,above] {$-3$} (n2)
	 			(n4) edge [->] node [left,above] {$5$}  (n2)
	 			(n4) edge [->] node [left,above] {$4$}  (n3);
	 			\end{tikzpicture}}
	 		
	 		\caption[Three System Representations]{The graphs of the system in Equation \eqref{eq:ex}. 
	 			~\ref{subfig:seq2}: $\G^{s}$ the signed graph. ~\ref{subfig:seq3}: $\G^{w}$ the weighted graph}
	 		\label{fig:ex}
	 	\end{figure}
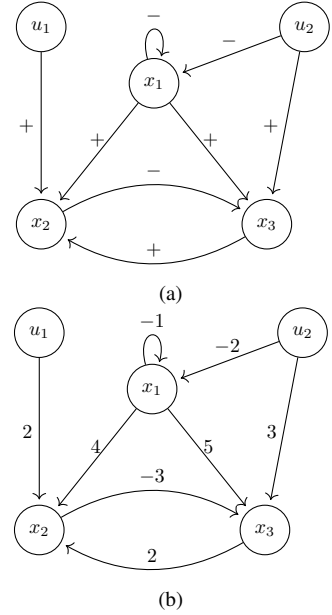
	 	
	 	A semi-walk from $v_{0}$ to $v_{k}$, $\pi_{s}(v_{0},v_{k})$, is a collection of nodes $v_{0},v_{1},v_{2}\dots, v_{k-1},v_{k}\in\mathcal{V}$, as well as $k$ edges which satisfy $(v_{i-1},v_{i})\in \mathcal{E} \vee (v_{i},v_{i-1})\in \mathcal{E}$. For convenience, the semi-walk can be represented by $\pi_{s}=v_{0},\hat{e}_{1},v_{1},\hat{e}_{2},v_{2}\dots, v_{k-1},\hat{e}_{k},v_{k}$ where $\hat{e}_{i}$ is the element of $\{(v_{i-1},v_{i}),(v_{i},v_{i-1})\}$ that is contained in $\mathcal{E}$. Like a walk, the sign of a semi-walk follows $s(\pi_{s})=\prod_{\hat{e}_{i}\in\pi_{s}} \ s(\hat{e}_{i})$ and the weight of a semi-walk  follows $w(\pi_{s})=\prod_{\hat{e}_{i}\in\pi_{s}} \ w(\hat{e}_{i})$. A semi-walk is a semi-path if the nodes of the semi-walk are distinct and a semi-walk is a semi-cycle if the first and last element of the semi-walk are the same.
	 	
	 	 \begin{definition}A graph is structurally balanced if all semi-cycles have a positive sign \cite{harary1953notion}.\end{definition} Structural balance has been shown to be related to the controllability of signed consensus dynamics evolving over a network \cite{alemzadeh2017controllability}. While no explicit relation between structural balance and herdability is shown in this paper, there is a relationship between structural balance and the controllability matrix $\C$ which will be explored in further depth later in the paper. 
	 	
	 	To begin classifying a linear system based on the signed graph $\G^{s}$, we define two basic types of sets. 
	 	Let $\nd^{j}$ be the set of nodes reachable from $v_{uj}$ via at least one negative walk of length $d$. Similarly $\pd^{j}$ is the set of nodes reachable from $v_{uj}$ through at least one positive walk of length $d$. If there is only one input to the system, the superscript will be dropped to refer to $\nd$ and $\pd$ instead of $\nd^{1}$ and $\pd^{1}$. {In Figure~\ref{fig:ex}, $\mathcal{N}_{1}^{1}=\{\emptyset\}$, $\mathcal{N}_{2}^{1}=\{x_{2}\}$,  $\mathcal{P}_{1}^{1}=\{x_{2}\}$,  $\mathcal{P}_{2}^{1}=\{\emptyset\}$} and so on.
	 	
	 	The sets $\pd^{j}$ and $\nd^{j}$ can provide sufficient information to determine the sign-herdability of a linear system. 
	 	Consider the total weight of positively signed walks from input $v_{uj}$ to node $v_{xi}$ with length $d$, \begin{equation}{\rho}_{j\rightarrow i,d}^{+}=\sum_{\pi\in \theta_{d}^{+}(v_{uj},v_{xi})}w(\pi),\end{equation}
	 	where $\theta_{d}^{+}(v_{uj},v_{xi})$ is the set of positive walks of length $d$ from $v_{uj}$ to $v_{xi}$. From the definition of $\pd^{j}$, it holds that ${\rho}_{j\rightarrow i,d}^{+}>0$ if $v_{xi}\in\pd^{j}$ and $0$ else. Similarly the total weight of negatively signed walks from input $v_{uj}$ to node $v_{xi}$ with length $d$ is \begin{equation}{\rho}_{j\rightarrow i,d}^{-}=\sum_{\pi\in \theta_{d}^{-}(v_{uj},v_{xi})}w(\pi),\end{equation} where $\theta_{d}^{-}(v_{uj},v_{xi})$ is the set of negative walks of length $d$ from $v_{uj}$ to $v_{xi}$ and it follows that ${\rho}_{j\rightarrow i,d}^{-}<0$ if $v_{xi}\in\nd^{j}$ and $0$ else. The weight of all walks from input $v_{uj}$ of length $d$ is \begin{equation}{\rho}_{j\rightarrow i,d}={\rho}_{j\rightarrow i,d}^{+}+{\rho}_{j\rightarrow i,d}^{-}. \end{equation}

	 	{We will be interested in cases where the sign of ${\rho}_{j\rightarrow i,d}$ is fixed across all possible weights on the edges, as this gives rise to a sign definite controllability matrix. This occurs if all paths at a certain distance are of the same sign.}
	 	\begin{proposition}\thlabel{lem:pdti}
	 		If $v_{xi}\in\pd^{j}\wedge v_{xi}\notin\nd^{j}$ then ${\rho}_{j\rightarrow i,d}>0$. {If $v_{xi}\in\nd^{j}\wedge v_{xi}\notin\pd^{j}$ then ${\rho}_{j\rightarrow i,d}<0$. Further if $v_{xk}\in \nd^{j}\cup\pd^{j} \wedge v_{xk}\notin \nd^{j}\cap\pd^{j}$ then ${\rho}_{j\rightarrow i,d}\neq 0$.}
	 	\end{proposition}
{
	\begin{proof}
		If $v_{xi}\in\pd^{j}\wedge v_{xi}\notin\nd^{j}$, then ${\rho}_{j\rightarrow i,d}^{-}=0$ and ${\rho}_{j\rightarrow i,d}$ is manifestly positive. The case $v_{xi}\in\nd^{j}\wedge v_{xi}\notin\pd^{j}$ follows similarly.
		
		To show that $v_{xk}\in \nd^{j}\cup\pd^{j} \wedge v_{xk}\notin \nd^{j}\cap\pd^{j}$ then ${\rho}_{j\rightarrow i,d}\neq 0$, suppose the contrary. Then
		\begin{align}
		{\rho}_{j\rightarrow k,d}&=0\\
		{\rho}_{j\rightarrow k,d}^{+}&+{\rho}_{j\rightarrow k,d}^{-}=0.
		\end{align}
		As $v_{xk}\in \nd^{j}\cup\pd^{j}$ it holds that
		\begin{align}
		{\rho}_{j\rightarrow k,d}^{+}>0, \ & {\rho}_{j\rightarrow k,d}^{-}<0
		\end{align} 
		which implies that
		\begin{align}
		v_{xi}\in\pd^{j}, \ & v_{xi}\in\nd^{j} \\
		v_{xi}\in\pd^{j}&\cap \nd^{j}
		\end{align} 
\end{proof} }
	 	
	 	It is possible to relate ${\rho}_{j\rightarrow i,d}$ with the system matrices $A,B$ and ultimately the controllability properties of the system.  
	 	Define a weighted adjacency matrix $\tilde{A}_{w}$ for $\G_{x}^{w}$, where $(\tilde{A}_{w})_{i,j}=w((v_{xj},v_{xi}))$ if $(v_{xj},v_{xi})\in\mathcal{E}_{x}$ and $(\tilde{A}_{w})_{i,j}=0$ if not. Define a weighted adjacency matrix $\tilde{B}_{w}$ for $\G_{u}^{w}$, where $(\tilde{B}_{w})_{i,j}=w((v_{uj},v_{xi}))$ if $(v_{uj},v_{xi})\in\mathcal{E}_{u}$ and $(\tilde{B}_{w})_{i,j}=0$ if not. Note that from the definition of the weight of an edge, $\tilde{A}_{w}=A$ and $\tilde{B}_{w}=B$. Then $(A^{d-1}B)_{i,j}$ is the sum of the weight of all walks of length $d$ from $v_{uj}$ to $v_{xi}$, {which can be used to show that $\mathcal{C}$ is determined by walks on $\G^{w}$ which have lengths from $1$ to $n$. } More formally:
	 	
	 	\begin{lemma}\thlabel{lem:cont} $\mathcal{C}_{i,(m(d-1)+j)}=\rho_{j\rightarrow i,d}$. \end{lemma}
	 	{
	 		\begin{proof}
	 			We first show that $(A^{d-1}B)_{i,j}={\rho}_{j\rightarrow i,d},$ via proof by induction on $d$. 
	 			Consider the case of $d=1$. By the definition of the weight of an edge: 
	 			\begin{equation}
	 			B_{i,j}={\rho}_{j\rightarrow i,1}.
	 			\end{equation}
	 			Consider the weight of all walks of length $d$ from an input $v_{uj}$ to a state node $v_{xi}$. By assumption, $(A^{d-2}B)_{i,j}={\rho}_{j\rightarrow i,d-1}$. As $A^{d-1}B=AA^{d-2}B$, it follows that \begin{equation}
	 			(A^{d-1}B)_{i,j}=\sum_{k=1}^{n} (A)_{i,k}{\rho}_{j\rightarrow k,d-1}.
	 			\end{equation} As a walk of length $d$ is the concatenation of a walk of length $d-1$ and a walk of length $1$, it follows from the definition of the weight of a walk that 
	 			\begin{equation}
	 			\sum_{k=1}^{n} A_{i,k}{\rho}_{j\rightarrow k,d-1}={\rho}_{j\rightarrow i,d}.
	 			\end{equation}
	 			The result follows directly from the fact that $\mathcal{C}_{i,(m(d-1)+j)}=(A^{d-1}B)_{i,j}.$
	 	\end{proof}
 	The definition of the matrix exponential, $e^{A}=\sum_{z=0}^{\infty} \frac{1}{k!} A^{k}$, and the form of the controllability grammian shows that the controllability grammian is related to all possible paths in the underlying graph of $A$. }

	 		{To conclude this section, conditions on the graph which produce sign definite controllability matrix and controllability grammian are discussed.	 	 }
	 
	 	\begin{theorem}\thlabel{th:signcon}
	 		If $\forall x_i\in\{x_1,x_2,\dots,x_n\}$ it holds that for each $d$ and $j$ such that $\mathcal{C}_{i,m*(d-1)+j} \neq0$, $x_i$ satisfies $v_{xi}\in\pd^{j}\wedge v_{xi}\notin\nd^{j}$  or $v_{xi}\in\nd^{j}\wedge v_{xi}\notin\pd^{j}$, then {the controllability matrix is sign definite}.
	 	\end{theorem}
	 	\begin{proof}
	 		If the condition of the theorem holds, then every non-zero element of $\mathcal{C}$ is associated only with paths of the same sign {by \thref{lem:pdti}} and as such will have always have the same sign no matter the weights on the graphs. 
	 	\end{proof}
	 	If the underlying graph of the system satisfies \thref{th:signcon}, the controllability matrix $\C$ will always have the same sign pattern, however the condition of the theorem {would require checking all paths of length $1$ to $n$ between input and state nodes. }  
	 	It is possible to show a stronger condition more easily. Consider the following:
	 	
	 	\begin{theorem}\thlabel{th:struc}
	 		If the system graph is structurally balanced, then {the controllability matrix is sign definite.} 
	 	\end{theorem}
	 	\begin{proof}
	 		 {Structural balance implies that the paths between an pair of nodes have the same sign \cite{cartwright1956structural}, which implies that $v_{xk}\in \nd^{j}\cup\pd^{j} \wedge v_{xk}\notin \nd^{j}\cap\pd^{j}$ for every non-zero element of $\C$ and $(A^{d-1}B)_{kj}$ must have a unique sign. If not,} there exists an $x_i$ such that there is a $d$ and $j$ where $v_{xi}\in\pd^{j}\wedge v_{xi}\in\nd^{j}$. This implies there are one or more positive paths of length $d$ from $j$ and one or more negative paths of length $d$ from $j$. Without loss of generality, consider one positive path and one negative path from input $u_j$ to node $x_i$. These paths form a semi-cycle in the graph. One of the paths is negative and must have an odd number of negative edges. The other is positive and must have an even number of negative edges. As such the semi-cycle must have an odd number of negative edges, i.e. the semi-cycle must have negative weight, which implies that the graph is not structurally balanced.  
	 	\end{proof}
	 	
	 	{\thref{th:struc} also shows that if the graph is structurally balanced then there can be no path cancellation such that ${\rho}_{j\rightarrow i,d}\neq 0$ by \thref{lem:pdti}. However structural balance is not a necessary condition for a sign definite controllability matrix. } Consider the graph in Figure \ref{fig:further}. Consider the node $v_{x2}$, which is in $\mathcal{N}_{1}$ and {$\mathcal{P}_{2}$} but not in $\mathcal{P}_{1}$ and {$\mathcal{N}_{2}$}. As such it satisfies the condition of \thref{th:signcon} and the controllability matrix $\C$ will always have the same sign, however the graph as a whole is not structurally balanced. As can be seen, structural balance ignores the lengths of the paths which connect an input to a state node which are important for determining herdability. 
	 	
	 	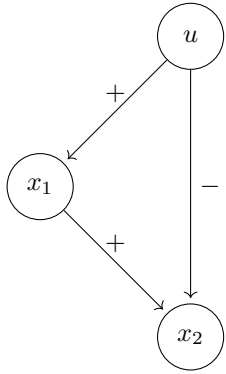
\begin{figure}[h]
	 		\centering
	 		\begin{tikzpicture}[
	 		->,
	 		shorten >=2pt,
	 		auto,
	 		node distance=0.5cm,
	 		every text node part/.style={align=center},
	 		scale=0.25, every node/.style={scale=1}
	 		]
	 		\tikzstyle{every state}=[fill=none,draw=black,text=black]
	 		\node[
	 		state,
	 		] (n0) at(0,0)
	 		{$u$};
	 		\node[
	 		state,
	 		] (n1) at(-8,-8)
	 		{$x_{1}$};
	 		\node[
	 		state,
	 		] (n2) at(0,-16)
	 		{$x_{2}$};
	 		\path (n0) edge [->] node [left,above] {$+$}(n1)
	 		(n0) edge [->]  node {$-$} (n2)
	 		(n1) edge [->] node[left,above] {{$+$}} (n2);
	 		\end{tikzpicture}
	 		\caption[Structural Balance and Sign Herdability]{An example of a graph which will always have the same sign pattern for the controllability matrix but is not structurally balanced.}
	 		\label{fig:further}
	 	\end{figure}
	 	
	 	However structural balance can be determined in linear time\cite{MAYBEE198399}, which implies that \thref{th:struc} makes it possible to characterize the sign herdability of a system from the controllability matrix with little extra computational cost.
	 	 	
	 	\section{Graphical Conditions for Complete Herdability}\label{sec:graph_con}
	 	This section considers the relationship between graph structure and herdability in a number of different contexts. First {herdability is shown for the case where the underlying graph is a directed outbranching, then } a necessary condition for complete herdability is examined before developing a number of sufficient conditions for complete herdability. These sufficient conditions extend the characterization of Section \ref{sec:class}.  
	 	\subsection{Directed Out-branchings}\label{sec:outb}
	 	 
	 	This section {considers herdability} in the special case of a system which has an underlying graph which is a rooted out-branching. {In such systems, it is possible to completely characterize the set of nodes that can be herded together, $\mathbb{H}_{u}$.} 
	 	
	 	A directed graph, $\hat{\mathcal{G}}=(\hat{\mathcal{V}},\hat{\mathcal{E}})$ is a rooted out-branching if it has a root node $v_{i}\in\hat{\mathcal{V}}$ such that for every other node $v_{j}\in\hat{\mathcal{V}}$ there is a single directed walk from $v_{i}$ to $v_{j}$. The case considered here is that of  a single input, input rooted out-branching, which means that every node $v_{xi}\in\hat{\mathcal{V}}_{x}$ has a single in-bound walk from the single input $v_{u}$. The unique walk from $v_{u}$ to $v_{xi}$ in the input-rooted out-branching will be referred to as $\pi_{t}(v_{u},v_{xi})$. Consider the maximum walk length between $v_{u}$ and a state node, which is  \begin{equation}d_{\max}=\max_{v_{xi}\in\hat{\mathcal{V}}_{x}} \mathrm{len}(\pi_{t}(v_{u},v_{xi})).\end{equation}
	 	  
	 	\begin{theorem}\thlabel{th:outtree}
	 		In an input rooted, out-branching, $\mathbb{H}_{u}$ follows 
	 		\begin{equation}
	 		\mathbb{H}_{u}=\bigcup_{d=1}^{d_{\max}} \mathcal{X}_{d},
	 		\end{equation} 
	 		where $\mathcal{X}_{d}\in\{\pd,\nd,\emptyset\}$. 
	 	\end{theorem}
	 	\begin{proof}
	 		Consider the ability to herd a node $v_{xi}$ and assume that $\mathrm{len}(\pi_{t}(v_{u},v_{xi}))=d_{i}$. As there is only one walk from $v_{u}$ to $v_{xi}$ it holds that $(\mathcal{C})_{i,d}=0, \ ~\forall d\neq d_{i}\in\{1,2,\dots,d_{\max}\},$ and $(\mathcal{C})_{i,d_{i}}\neq0$. Further $v_{xi}$ is either in $\pd$ or in $\nd$ but can not be in both as there is only one path to $v_{xi}$. Then {by \thref{lem:pdti} and \thref{lem:cont}} if $v_{xi}$ is in $\pd[d_{i}]$, ${\rho}_{u\rightarrow i,d}>0$ and consequently $(\mathcal{C})_{i,d_{i}}>0$  or if $v_{xi}$ is in $\nd[d_{i}]$, ${\rho}_{u\rightarrow i,d}<0$ and $(\mathcal{C})_{i,d_{i}}<0$. 
	 		
	 		Then it follows that $(\mathcal{C})_{:, d_{i}}$ uniquely determines the ability to herd all nodes at distance $d_{i}$. Consider the choice of the constant $\alpha_{di}$. If $\alpha_{di}=1$ then $((\mathcal{C})_{:, d_{i}}\alpha_{di})_{i}>0, \ ~\forall i \text{ such that } v_{xi}\in\pd[d_{i}]$ and $\pd[d_{i}]$ is herdable by \thref{lem:herdset}. If $\alpha_{di}=-1$ then  $((\mathcal{C})_{:, d_{i}}\alpha_{di})_{j}>0, \ ~\forall i \text{ such that } v_{xi}\in\nd[d_{i}]$ and $\nd[d_{i}]$ is herdable by \thref{lem:herdset}. Finally if $\alpha_{di}=0$ then  $(\mathcal{C})_{:, d_{i}}\alpha_{di}=\mathbf{0}_{n}$ and no nodes are herded. Then by the appropriate choice of $\alpha_{di}$ the set of nodes that can be herded at distance $d_{i}$ from $u$, $\mathcal{X}_{di}$ must be one of $\{\pd,\nd,\emptyset\}$.
	 		
	 		Construct a vector $\pmb{\alpha}\in\mathbb{R}^{n}$ where $\forall d\in\{1,2,\dots,d_{\text{max}}\} $
	 		\begin{equation} \pmb{\alpha}_{d}=\begin{cases}1 \text{ \ \ so that } \mathcal{X}_d=\pd,\\ -1 \text{ so that } \mathcal{X}_d=\nd, \\ 0 \text{\ \ \ so that } \mathcal{X}_d=\emptyset,
	 		\end{cases}
	 		\end{equation} and where the remaining $n-d_{\max}$ elements are $0$. Then $\mathcal{C}\pmb{\alpha}$ shows the herdability of the set of nodes $\bigcup_{d=1}^{d_{\max}} \mathcal{X}_{d}$. 
	 	\end{proof}
	 	\begin{corollary}\thlabel{th:size}
	 		The maximal collection of nodes, $\mathbb{H}_{u}^{\ast}$, that can be herded in a input rooted out-branching satisfies \begin{equation}
	 		|\mathbb{H}_{u}^{\ast}|=\sum_{l=1}^{d_\text{max}}\max(|\mathcal{N}_{l}|,|\mathcal{P}_{l}|).
	 		\end{equation}
	 		
	 	\end{corollary}
	 	\ \\
	 	
	 	In the case of an single input, input connectable,  directed out-branching where $\forall d \in \{1,2,\dots,d_{\text{max}}\}, \ \mathcal{N}_{d}=\emptyset \veebar \mathcal{P}_{d}=\emptyset$, \thref{th:size} shows that $|\mathbb{H}_{u}^{\ast}|=n$, or equivalently that the system is completely herdable. Figure~\ref{fig:choose} shows an example of selecting the set of nodes that can be herded in an input rooted, out-branching. 
	 	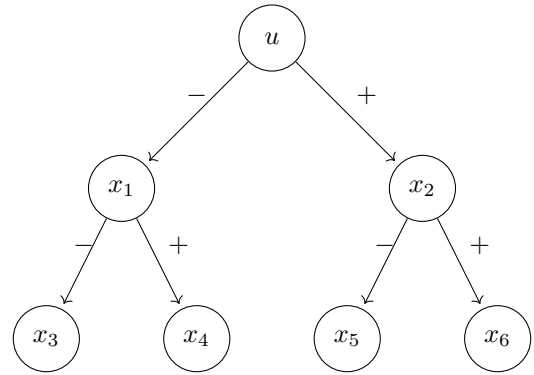
\begin{figure}[h]
	 		\centering
	 		\begin{tikzpicture}[
	 		->,
	 		shorten >=2pt,
	 		auto,
	 		node distance=0.5cm,
	 		every text node part/.style={align=center},
	 		scale=0.25, every node/.style={scale=1}
	 		]
	 		\tikzstyle{every state}=[fill=none,draw=black,text=black]
	 		\node[
	 		state,
	 		] (n0) at(0,0)
	 		{$u$};
	 		\node[
	 		state,
	 		] (n1) at(-8,-8)
	 		{$x_{1}$};
	 		\node[
	 		state,
	 		] (n2) at(8,-8)
	 		{$x_{2}$};
	 		\node[
	 		state,
	 		] (n3) at(4,-16)
	 		{$x_{5}$};
	 		\node[
	 		state,
	 		] (n4) at(12,-16)
	 		{$x_{6}$};
	 		\node[
	 		state,
	 		] (n5) at(-12,-16)
	 		{$x_{3}$};
	 		\node[
	 		state,
	 		] (n6) at(-4,-16)
	 		{$x_{4}$};
	 		\path (n0) edge [->] node [left,above] {$-$}(n1)
	 		(n0) edge [->]  node {$+$} (n2)
	 		(n2) edge [->] node[left,above] {$-$} (n3)
	 		(n2) edge [->]  node {$+$} (n4)
	 		(n1) edge [->] node[left,above] {$-$} (n5)
	 		(n1) edge [->]  node {$+$} (n6);
	 		\end{tikzpicture}
	 		\caption[Input Rooted Out-branching]{An example of an input rooted out-branching}
	 		\label{fig:choose}
	 	\end{figure}
	 	
	 	The graph in Figure~\ref{fig:choose} can be translated into the following class of systems:
	 	
	 	\begin{equation}
	 	\dot{\mathbf{x}}=\begin{bmatrix}
	 	0 & 0 & 0 & 0 & 0 & 0 \\ 
	 	0 & 0 & 0 & 0 & 0 & 0 \\ 
	 	-\alpha_{1} & 0 & 0 & 0 & 0 & 0 \\ 
	 	\alpha_{2} & 0 & 0 & 0 & 0 & 0 \\ 
	 	0 & -\alpha_{3} & 0 & 0 & 0 & 0 \\ 
	 	0 & \alpha_{4} & 0 & 0 & 0 & 0
	 	\end{bmatrix}\mathbf{x}+ \begin{bmatrix}
	 	-\beta_{1}  \\ 
	 	\beta_{2}  \\ 
	 	0  \\ 
	 	0  \\ 
	 	0  \\ 
	 	0 
	 	\end{bmatrix}\mathbf{u} 
	 	\end{equation}
	 	where $\alpha_{1},\alpha_{2},\alpha_{3},\alpha_{4},\beta_{1},\beta_{2}>0$. The system has a controllability matrix:
	 	\begin{equation*}
	 	\mathcal{C}=\begin{bmatrix}
	 	-\beta_{1} & 0 & 0 & 0 & 0 & 0 \\ 
	 	\beta_{2} & 0 & 0 & 0 & 0 & 0 \\ 
	 	0& \alpha_{1}\beta_{1} & 0 & 0 & 0 & 0 \\ 
	 	0& -\alpha_{2}\beta_{1} & 0 & 0 & 0 & 0 \\ 
	 	0 & -\alpha_{3}\beta_{2} & 0 & 0 & 0 & 0 \\ 
	 	0 & \alpha_{4}\beta_{2} &0  & 0 & 0 & 0
	 	\end{bmatrix} 
	 	\end{equation*}
	 	where
	 	\begin{equation}
	 	\mathrm{range}(\mathcal{C})=span\left(\left\{ \begin{bmatrix}
	 	-\beta_{1}  \\ 
	 	\beta_{2}  \\ 
	 	0  \\ 
	 	0  \\ 
	 	0  \\ 
	 	0 
	 	\end{bmatrix} , \begin{bmatrix}
	 	0  \\ 
	 	0  \\ 
	 	\alpha_{1}\beta_{1}  \\ 
	 	-\alpha_{2}\beta_{1}  \\ 
	 	-\alpha_{3}\beta_{2}  \\ 
	 	\alpha_{4}\beta_{2} 
	 	\end{bmatrix} \right\}\right)
	 	\end{equation}
	 	As such the possible sets of herded nodes are $\{1,3,6\},\{1,4,5\},\{2,3,6\},\{2,4,5\}$.
	 	
	 	The result of \thref{th:outtree} is similar in nature to the $k$-walk controllability theory\cite{Gao2014a}. The $k$-walk theory shows that for each $d\in\{1,2,\dots,d_{\text{max}}\}$ one {node} can be controlled. In the graph given in Figure~\ref{fig:choose}, the possible sets of nodes that can be controlled are $\{1,3\},\{1,4\},\{1,5\},\{1,6\},\{2,3\},\{2,4\}\{2,5\},\{2,6\}.$ As a consequence of the k-walk theory, the maximal collection of nodes that are controlled in a directed out-branching from input $v_{u}$, $\mathbb{C}_{u}^{\ast}$, satisfies \begin{equation}
	 	|\mathbb{C}_{u}^{\ast}|= d_{\text{max}}.
	 	\end{equation} In the case of herding a network, \thref{th:size} shows that the maximal collection of nodes, $\mathbb{H}_{u}^{\ast}$, will satisfy
	 	\begin{equation}
	 	d_{\text{max}} \leq|\mathbb{H}_{u}^{\ast}|\leq n.
	 	\end{equation}
	 	Therefore in the worst case, the same number of nodes can be herded as can be controlled and depending on the network structure many more nodes can be herded. 
	 	
	 	\subsection{A Necessary Condition for Complete Herdability}\label{sec:nes}
	 	
	 	This section shows how graph structure and system herdability are related by providing a necessary condition for complete herdability of a system known as input connectability. Note that input connectability is also a necessary condition for structural controllability \cite{Lin1974} and sign controllability \cite{tsatsomeros1998sign}. 
	 	
	 	\begin{definition}A graph is input connectable (equivalently, accessible) if \begin{equation} \bigcup_{v_{uj}\in\mathcal{V}_u} \mathscr{R}_{j}=\mathcal{V}_{x}, \end{equation} where $\mathscr{R}_{j}$ is the set of nodes reachable from $v_{uj}$:  $\mathscr{R}_{j}=\{v_{xi}\in\mathcal{V}_{x} \ | \ v_{uj} \rightarrow v_{xi}\}$.\end{definition}
	 	
	 	If a single node is not {accessible} then the system is not completely herdable. 
	 	
	 	\begin{theorem} \thlabel{th:innec}
	 		If a system is completely herdable, then it is input connectable. 
	 	\end{theorem}
	 	\begin{proof}
	 		Suppose that there exists a node $v_{xi}$ such that $v_{xi}\notin\bigcup_{j}\mathscr{R}_{j}$ and as such there is no walk from an input to $v_{xi}$, i.e. the system is not input connectable.
	 		If there is no walk to $v_{xi}$, then $(\mathcal{C})_{i,:}=\mathbf{0}_n$ by \thref{lem:cont} and the node will not be herdable. {To see why consider making $\mathbf{x}(t)\geq h$ from an initial state $\mathbf{x}(0)=\mathbf{0}_{n}$. As $\forall j \ (\mathcal{C})_{i,j}=0$, it holds that $\forall \mathbf{z}\in \mathrm{range}(\mathcal{C}), \ \mathbf{z}_{i}=0$ and by \thref{lem:rceq} for any reachable $\mathbf{x}(t) \ ~\forall t\geq 0$, $\mathbf{x}_{i}(t)=0$ and state $x_i$ is not herdable.} As such, the system is not completely herdable. 
	 	\end{proof}
	 	
	 	{As was just seen in the case of the directed out-branching in Section \ref{sec:outb}, the underlying sign pattern influences herdability and as such input connectability is only a necessary condition. In the special case that the system is positive, input connectability is a sufficient condition as well. } 
	 	
	 	\begin{theorem}\thlabel{th:posin}
	 		A positive linear system is completely herdable if and only if it is input connectable. 
	 	\end{theorem}
	 	\begin{proof}
	 		(Sufficiency)
	 		By Theorem $8$ of \cite{Farina2011}, an input connectable, positive linear system is excitable. Then there is an element-wise positive vector in the reachable subspace, which is also the range of the controllability matrix by \thref{lem:rceq}. Then by \thref{cor:krange}, the system is completely herdable. 
	 		
	 		\noindent(Necessity) Follows from \thref{th:innec}.
	 	\end{proof}
	 	{\thref{th:posin} is important for the application of herdability as many biological and social systems are described by positive systems \cite{Farina2011}. 
	 		
	 		Another possibility is that for a given sign pattern, paths cancel in the graph which would prevent an input connectible graph from being completely herdable. As discussed at the end of Section \ref{sec:pre}, if the underlying graph is structurally balanced then such cancellations can not occur. We conjecture that structural balance and accessibility may form the basis for a sufficient condition for herdability. }
	 	
	 	\subsection{Sufficient Graph Conditions for Sign Herdability}\label{sec:suf_graph}
	 	This section will now consider the sufficient condition of Section \ref{sec:class} in light of the characterization of the controllability matrix given in Section~\ref{sec:pre}. The following Theorems provide a case where the composition of the sets $\pd^{j}$ and $\nd^{j}$ uniquely determines the sign of the columns of the controllability matrix and in turn the herdability of the graph, i.e. one is able to show the sign-herdability of the system. {These conditions are also less stringent than the requirements for a sign definite $\C$ as they only ask that certain columns of the $\C$ have a fixed sign pattern.}  
	 	\begin{definition}
	 		A node is \emph{sign strictly herdable} if it is strictly herdable for all graphs with the same sign pattern.
	 	\end{definition}

 	\begin{lemma}\thlabel{lem:ssh}
If for the node $v_{xi}$ there exists a distance $d$ and an input $v_{uj}$ such that $v_{xi}\in\mathcal{N}_{d}^{j}\cup\mathcal{P}_{d}^{j}$ and $\mathcal{N}_{d}^{j}=\emptyset \veebar \mathcal{P}_{d}^{j}=\emptyset$, where $\veebar$ denotes exclusive OR, it is sign strictly herdable. 
 	\end{lemma}
 \begin{proof}
	 		Consider the herdability of a node $v_{xi}$ which satisfies $v_{xi}\in\mathcal{N}_{d^{i}}^{j^{i}}\cup\mathcal{P}_{d^{i}}^{j^{i}}$ and $\mathcal{N}_{d^{i}}^{j^{i}}=\emptyset \veebar \mathcal{P}_{d^{i}}^{j^{i}}=\emptyset$ for some $d^{i}$ and $v_{uj^{i}}$. The fact that $\mathcal{N}_{d^{i}}^{j^{i}}=\emptyset \veebar \mathcal{P}_{d^{i}}^{j^{i}}=\emptyset$ implies that $\mathcal{N}_{d^{i}}^{j^{i}}\cap\mathcal{P}_{d^{i}}^{j^{i}}=\emptyset$, and as such it must be that $v_{xi}\in\mathcal{N}_{d^{i}}^{j^{i}}\cup\mathcal{P}_{d^{i}}^{j^{i}}$ and $v_{xi}\notin\mathcal{N}_{d^{i}}^{j^{i}}\cap\mathcal{P}_{d^{i}}^{j^{i}}$.

From \thref{lem:cont} and \thref{lem:pdti}, this implies $(\mathcal{C})_{i,m(d^{i}-1)+j^{i}}\neq 0$. 
Additionally, as $\mathcal{N}_{d^{i}}^{j^{i}}=\emptyset \veebar \mathcal{P}_{d^{i}}^{j^{i}}=\emptyset$, \thref{lem:pdti} show that all nonzero elements of $(\mathcal{C})_{:,m(d^{i}-1)+j^{i}}$ have the same sign and that the sign does not depend on the edge weights. Therefore for all graphs with the same sign pattern, the node is contained in $\mathcal{S}$, the set of nodes for which there exists a unisigned column of $\C$. Then by \thref{lem:s} each node is sign strictly herdable.
 \end{proof}
{For simple graph structures, the condition of \thref{lem:ssh} can be easily shown to be satisfied. For example in the case of directed outbranching, the condition holds if the columns of $A$ and $B$ are unisigned. which corresponds to each layer of the out branching having connections of the same sign.  }
	 	\begin{theorem}\thlabel{th:all}
	 		If for each $v_{xi}\in\mathcal{V}_{x}$ there exists a distance $d$ and an input $v_{uj}$ such that $v_{xi}\in\mathcal{N}_{d}^{j}\cup\mathcal{P}_{d}^{j}$ and $\mathcal{N}_{d}^{j}=\emptyset \veebar \mathcal{P}_{d}^{j}=\emptyset$, then the system is completely sign herdable. 
	 	\end{theorem}
	 	\begin{proof}
 By \thref{lem:ssh} each node is sign strictly herdable and the system is completely sign herdable by \thref{th:suff}.
	 	\end{proof}
	 	\thref{th:all} is an extension of \thref{lem:s} to the sign structure of the network. Consider the following definition which allows the extension of \thref{lem:shat}.
	 	
	 	\begin{definition}
	 		A node $v_{xi}$ is said to be sign balanced if there exists a distance $d$ and an input $v_{uj}$ such that $v_{xi}\in\mathcal{N}_{d}^{j}\cup\mathcal{P}_{d}^{j}$, $v_{xi}\notin\mathcal{N}_{d}^{j}\cap\mathcal{P}_{d}^{j}$,  and all nodes that balance $v_{xi}$ at distance $d$ from an input $v_{uj}$ are sign strictly herdable.
	 	\end{definition}
	 	\begin{theorem} \thlabel{th:allplus}
	 		If all nodes are sign balanced then the system is completely sign herdable. 
	 	\end{theorem}
	 	\begin{proof}
	 		As for each $v_{xi}$ there exists a distance $d$ and an input $v_{uj}$ such that $v_{xi}\in\mathcal{N}_{d}^{j}\cup\mathcal{P}_{d}^{j}$, $v_{xi}\notin\mathcal{N}_{d}^{j}\cap\mathcal{P}_{d}^{j}$, there is a column of $\mathcal{C}$ whose sign with respect to $v_{xi}$ is always consistent regardless of the weight of the edges in the walks that connect the input $v_{uj}$ and $v_{xi}$. As it is balanced by sign strictly herdable nodes, node $v_{xi}$ is sign strictly herdable. As all nodes are strictly sign herdable then the whole system is sign herdable by \thref{th:suff}. 
	 	\end{proof}
	 	\thref{th:all} and \thref{th:allplus}, as well as \thref{lem:s,lem:shat} (their counterpart based on the controllability matrix $\mathcal{C}$), provide sufficient conditions to verify that a node is strictly herdable. However as they are only sufficient there are completely herdable systems which can not be identified by verifying the conditions of \thref{lem:s,lem:shat,th:all,th:allplus}. Figure~\ref{fig:exnond} shows a simple example.
	 	\begin{figure}[h]
	 		\centering
	 		\begin{tikzpicture}[
	 		->,
	 		shorten >=2pt,
	 		auto,
	 		node distance=0.5cm,
	 		every text node part/.style={align=center},
	 		scale=0.25, every node/.style={scale=1}
	 		]
	 		\tikzstyle{every state}=[fill=none,draw=black,text=black]
	 		\node[
	 		state,
	 		] (n0) at(0,0)
	 		{$u$};
	 		\node[
	 		state,
	 		] (n1) at(-12,-12)
	 		{$x_{1}$};
	 		\node[
	 		state,
	 		] (n2) at(0,-12)
	 		{$x_{2}$};
	 		\node[
	 		state,
	 		] (n3) at(12,-12)
	 		{$x_{3}$};
	 		\path (n0) edge [->] node [left,above] {$+$}(n1)
	 		(n0) edge [->]  node {$+$} (n2)
	 		(n0) edge [->]  node {$-$} (n3)
	 		(n3) edge [->] node[left,above] {$+$} (n2)
	 		(n2) edge [->]  node [above] {$+$} (n1)
	 		(n1) edge [->, bend right]  node [below] {$+$} (n3);
	 		\end{tikzpicture}
	 		\caption[A Counter Example]{An example of a completely herdable graph which can not be identified by inspecting the columns of $\mathcal{C}$ nor the sets $\mathcal{N}_{d}^{j}$ and $\mathcal{P}_{d}^{j}$}
	 		\label{fig:exnond}
	 	\end{figure}
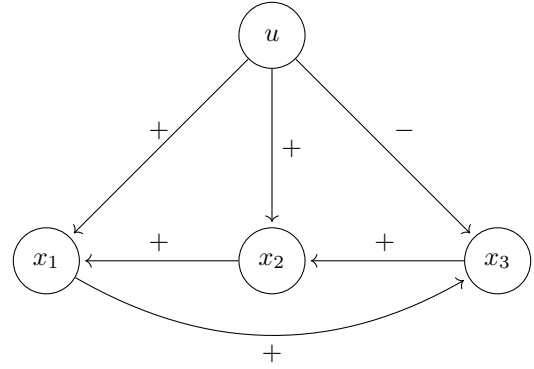
	 	
	 	\section{Conclusion}\label{sec:conc}
	 	In this paper, the basic theory of herdable systems was presented. The definition of herdability was shown to translate to simple conditions based on three matrices: the controllability grammian $\W$, the controllability matrix $\C$, and the matrix $\AB$. A number of sufficient conditions where shown which could be used to show herdability of a system by inspecting the columns of each of the aforementioned matrices. 
	 	
	 	The characterization of herdability based on the controllability matrix was extended to consider the underlying graph of the dynamical system. It was shown that a certain loss of symmetry, as shown by a balanced vector in the range of the controllability subspace, ensures that a system is no longer completely herdable. Additionally it was shown that as herdability is only dependent on the sign of an link between two state nodes, any characterization of herdability based on the controllability matrix can be extended to a class of systems with the same sign pattern if it can shown that the sign of the columns of the controllability matrix do not depend on the underlying edge weights.
	 	
	 		The results on the sign herdability of linear systems depend on the characterization of the sets $\mathcal{N}_{d}^{j}$ and $\mathcal{P}_{d}^{j}$. In the case where the underlying graph of $A$ is structurally balanced these need not be explicitly calculated as the sign herdability of the system can be determined directly from the controllability matrix $\C$. In the case where the underlying graph is not structurally balanced the sets $\mathcal{N}_{d}^{j}$ and $\mathcal{P}_{d}^{j}$ must be determined from the graph structure. 
	 		Unfortunately determining $\mathcal{N}_{d}^{j}$ and $\mathcal{P}_{d}^{j}$ via graph traversal involves considering all paths between an input node and a state node in the graph, of which there are potentially an exponential number.
	 	
	 	In the case where $A$ has an underlying graph which is a directed acyclic graph then these sets can be determined in linear time. If A is not acyclic then there is a possibility that the time complexity of the algorithm grows exponentially in the number of state nodes. As such the presented results, in terms of practical implementation, are currently best suited for graphs which are acyclic or sufficiently close to acyclic that the sets $\mathcal{N}_{d}^{j}$ and $\mathcal{P}_{d}^{j}$ can be calculated. 
\bibliographystyle{IEEEtran}
\bibliography{herdability}

\begin{thebibliography}{10}
\providecommand{\url}[1]{#1}
\csname url@samestyle\endcsname
\providecommand{\newblock}{\relax}
\providecommand{\bibinfo}[2]{#2}
\providecommand{\BIBentrySTDinterwordspacing}{\spaceskip=0pt\relax}
\providecommand{\BIBentryALTinterwordstretchfactor}{4}
\providecommand{\BIBentryALTinterwordspacing}{\spaceskip=\fontdimen2\font plus
\BIBentryALTinterwordstretchfactor\fontdimen3\font minus
  \fontdimen4\font\relax}
\providecommand{\BIBforeignlanguage}[2]{{%
\expandafter\ifx\csname l@#1\endcsname\relax
\typeout{** WARNING: IEEEtran.bst: No hyphenation pattern has been}%
\typeout{** loaded for the language `#1'. Using the pattern for}%
\typeout{** the default language instead.}%
\else
\language=\csname l@#1\endcsname
\fi
#2}}
\providecommand{\BIBdecl}{\relax}
\BIBdecl

\bibitem{kalmancont}
R.~E. Kalman, Y.~C. Ho, and K.~S. Narendra, ``Controllability of linear
  dynamical systems,'' \emph{Contributions to Differential Equations 1}, 1963.

\bibitem{mesbahi2010graph}
M.~Mesbahi and M.~Egerstedt, \emph{Graph Theoretic Methods in Multiagent
  Networks}.\hskip 1em plus 0.5em minus 0.4em\relax Princeton University Press,
  2010.

\bibitem{Liu2011}
Y.-Y. Liu, J.-J. Slotine, and A.-L. Barab{\'a}si, ``Controllability of complex
  networks,'' \emph{Nature}, vol. 473, no. 7346, pp. 167--173, 2011.

\bibitem{hespanha2009linear}
J.~P. Hespanha, \emph{Linear Systems Theory}.\hskip 1em plus 0.5em minus
  0.4em\relax Princeton University Press, 2009.

\bibitem{jadbabaie2019minimal}
A.~Jadbabaie, A.~Olshevsky, G.~J. Pappas, and V.~Tzoumas, ``Minimal
  reachability is hard to approximate,'' \emph{IEEE Transactions on Automatic
  Control}, vol.~64, no.~2, pp. 783--789, 2019.

\bibitem{Wu2014}
F.-X. Wu, L.~Wu, J.~Wang, J.~Liu, and L.~Chen, ``Transittability of complex
  networks and its applications to regulatory biomolecular networks,''
  \emph{Scientific Reports}, vol.~4, 2014.

\bibitem{Gao2014a}
J.~Gao, Y.-Y. Liu, R.~M. D'Souza, and A.-L. Barab{\'a}si, ``Target control of
  complex networks,'' \emph{Nature Communications}, vol.~5, 2014.

\bibitem{wonham}
W.~M. Wonham, \emph{Linear Multivariable Control: a Geometric Approach}.\hskip
  1em plus 0.5em minus 0.4em\relax Springer-Verlag, 1985.

\bibitem{egline}
M.~Egerstedt and C.~Martin, \emph{Control Theoretic Splines: Optimal Control,
  Statistics, and Path Planning}.\hskip 1em plus 0.5em minus 0.4em\relax
  Princeton University Press, 2009, vol.~29.

\bibitem{borrelli2017predictive}
F.~Borrelli, A.~Bemporad, and M.~Morari, \emph{Predictive Control for Linear
  and Hybrid Systems}.\hskip 1em plus 0.5em minus 0.4em\relax Cambridge
  University Press, 2017.

\bibitem{granovetter1978threshold}
M.~Granovetter, ``Threshold models of collective behavior,'' \emph{American
  Journal of Sociology}, vol.~83, no.~6, pp. 1420--1443, 1978.

\bibitem{schelling1971dynamic}
T.~C. Schelling, ``Dynamic models of segregation,'' \emph{Journal of
  Mathematical Sociology}, vol.~1, no.~2, pp. 143--186, 1971.

\bibitem{valente1995network}
T.~W. Valente, \emph{Network Models of the Diffusion of Innovations.}\hskip 1em
  plus 0.5em minus 0.4em\relax Cresskill New Jersey Hampton Press, 1995.

\bibitem{miller2001quorum}
M.~B. Miller and B.~L. Bassler, ``Quorum sensing in bacteria,'' \emph{Annual
  Reviews in Microbiology}, vol.~55, no.~1, pp. 165--199, 2001.

\bibitem{hodgkin1952quantitative}
A.~L. Hodgkin and A.~F. Huxley, ``A quantitative description of membrane
  current and its application to conduction and excitation in nerve,''
  \emph{The Journal of Physiology}, vol. 117, no.~4, pp. 500--544, 1952.

\bibitem{liu2015control}
Y.-Y. Liu and A.-L. Barab\'asi, ``Control principles of complex systems,''
  \emph{Rev. Mod. Phys.}, vol.~88, p. 035006, Sep 2016.

\bibitem{abelson1964mathematical}
R.~P. Abelson, ``Mathematical models of the distribution of attitudes under
  controversy,'' \emph{Contributions to Mathematical Psychology}, vol.~14, pp.
  1--160, 1964.

\bibitem{rahmani2009controllability}
A.~Rahmani, M.~Ji, M.~Mesbahi, and M.~Egerstedt, ``Controllability of
  multi-agent systems from a graph-theoretic perspective,'' \emph{SIAM Journal
  on Control and Optimization}, vol.~48, no.~1, pp. 162--186, 2009.

\bibitem{martini2010controllability}
S.~Martini, M.~Egerstedt, and A.~Bicchi, ``Controllability analysis of
  multi-agent systems using relaxed equitable partitions,'' \emph{International
  Journal of Systems, Control and Communications}, vol.~2, no. 1-3, pp.
  100--121, 2010.

\bibitem{nabi2013controllability}
M.~Nabi-Abdolyousefi and M.~Mesbahi, ``On the controllability properties of
  circulant networks,'' \emph{IEEE Transactions on Automatic Control}, vol.~58,
  no.~12, pp. 3179--3184, 2013.

\bibitem{parlangeli2012reachability}
G.~Parlangeli and G.~Notarstefano, ``On the reachability and observability of
  path and cycle graphs,'' \emph{IEEE Transactions on Automatic Control},
  vol.~57, no.~3, pp. 743--748, 2012.

\bibitem{notarstefano2013controllability}
G.~Notarstefano and G.~Parlangeli, ``Controllability and observability of grid
  graphs via reduction and symmetries,'' \emph{IEEE Transactions on Automatic
  Control}, vol.~58, no.~7, pp. 1719--1731, 2013.

\bibitem{chapman2014symmetry}
A.~Chapman and M.~Mesbahi, ``On symmetry and controllability of multi-agent
  systems,'' in \emph{IEEE 53rd Annual Conference on Decision and Control
  (CDC)}.\hskip 1em plus 0.5em minus 0.4em\relax IEEE, 2014, pp. 625--630.

\bibitem{Lin1974}
C.~T. Lin, ``Structural controllability,'' \emph{IEEE Transactions on Automatic
  Control}, vol.~19, no.~3, pp. 201--208, 1974.

\bibitem{shields1976structural}
R.~Shields and J.~Pearson, ``Structural controllability of multiinput linear
  systems,'' \emph{IEEE Transactions on Automatic Control}, vol.~21, no.~2, pp.
  203--212, 1976.

\bibitem{glover1976characterization}
K.~Glover and L.~Silverman, ``Characterization of structural controllability,''
  \emph{IEEE Transactions on Automatic Control}, vol.~21, no.~4, pp. 534--537,
  1976.

\bibitem{mayeda1979strong}
H.~Mayeda and T.~Yamada, ``Strong structural controllability,'' \emph{SIAM
  Journal on Control and Optimization}, vol.~17, no.~1, pp. 123--138, 1979.

\bibitem{johnson1993sign}
C.~R. Johnson, V.~Mehrmann, and D.~D. Olesky, ``Sign controllability of a
  nonnegative matrix and a positive vector,'' \emph{SIAM Journal on Matrix
  Analysis and Applications}, vol.~14, no.~2, pp. 398--407, 1993.

\bibitem{hartung2013characterization}
C.~Hartung, G.~Reissig, and F.~Svaricek, ``Characterization of sign
  controllability for linear systems with real eigenvalues,'' in \emph{3rd
  Australian Control Conference}.\hskip 1em plus 0.5em minus 0.4em\relax IEEE,
  2013, pp. 450--455.

\bibitem{tsatsomeros1998sign}
M.~J. Tsatsomeros, ``Sign controllability: Sign patterns that require complete
  controllability,'' \emph{SIAM Journal on Matrix Analysis and Applications},
  vol.~19, no.~2, pp. 355--364, 1998.

\bibitem{olesky1993qualitative}
D.~D. Olesky, M.~Tsatsomeros, and P.~Van Den~Driessche, ``Qualitative
  controllability and uncontrollability by a single entry,'' \emph{Linear
  Algebra and its Applications}, vol. 187, pp. 183--194, 1993.

\bibitem{quirk1965qualitative}
J.~Quirk and R.~Ruppert, ``Qualitative economics and the stability of
  equilibrium,'' \emph{The Review of Economic Studies}, vol.~32, no.~4, pp.
  311--326, 1965.

\bibitem{may1973qualitative}
R.~M. May, ``Qualitative stability in model ecosystems,'' \emph{Ecology},
  vol.~54, no.~3, pp. 638--641, 1973.

\bibitem{blanchini2014structural}
F.~Blanchini, E.~Franco, and G.~Giordano, ``A structural classification of
  candidate oscillatory and multistationary biochemical systems,''
  \emph{Bulletin of mathematical biology}, vol.~76, no.~10, pp. 2542--2569,
  2014.

\bibitem{brualdi2009matrices}
R.~A. Brualdi and B.~L. Shader, \emph{Matrices of Sign-Solvable Linear
  Systems}.\hskip 1em plus 0.5em minus 0.4em\relax Cambridge University Press,
  2009, vol. 116.

\bibitem{klee1984signsolvability}
V.~Klee, R.~Ladner, and R.~Manber, ``Signsolvability revisited,'' \emph{Linear
  algebra and its applications}, vol.~59, pp. 131--157, 1984.

\bibitem{wasserman1994social}
S.~Wasserman and K.~Faust, \emph{Social Network Analysis: Methods and
  Applications}.\hskip 1em plus 0.5em minus 0.4em\relax Cambridge University
  Press, 1994, vol.~8.

\bibitem{easley2010networks}
D.~Easley and J.~Kleinberg, \emph{Networks, Crowds, and Markets: Reasoning
  About a Highly Connected World}.\hskip 1em plus 0.5em minus 0.4em\relax
  Cambridge University Press, 2010.

\bibitem{altafini2013consensus}
C.~Altafini, ``Consensus problems on networks with antagonistic interactions,''
  \emph{IEEE Transactions on Automatic Control}, vol.~58, no.~4, pp. 935--946,
  2013.

\bibitem{alemzadeh2017controllability}
S.~Alemzadeh, M.~H. de~Badyn, and M.~Mesbahi, ``Controllability and
  stabilizability analysis of signed consensus networks,'' in \emph{Control
  Technology and Applications (CCTA), 2017 IEEE Conference on}.\hskip 1em plus
  0.5em minus 0.4em\relax IEEE, 2017, pp. 55--60.

\bibitem{rufacc}
S.~F. Ruf, M.~Egerstedt, and J.~S. Shamma, ``Herdable systems over signed,
  directed graphs,'' in \emph{2018 Annual American Control Conference
  (ACC)}.\hskip 1em plus 0.5em minus 0.4em\relax IEEE, 2018, pp. 1807--1812.

\bibitem{shao1999sign}
J.-Y. Shao, ``On sign inconsistent linear systems,'' \emph{Linear algebra and
  its applications}, vol. 296, no. 1-3, pp. 245--257, 1999.

\bibitem{klee1987recursive}
V.~Klee, ``Recursive structure of s-matrices and an o (m2) algorithm for
  recognizing strong sign solvability,'' \emph{Linear Algebra and its
  Applications}, vol.~96, pp. 233--247, 1987.

\bibitem{harary1953notion}
F.~Harary, ``On the notion of balance of a signed graph,'' \emph{The Michigan
  Mathematical Journal}, vol.~2, no.~2, pp. 143--146, 1953.

\bibitem{harary2005structural}
------, \emph{Structural Models: An Introduction to the Theory of Directed
  Graphs}.\hskip 1em plus 0.5em minus 0.4em\relax John Wiley \& Sons Inc.,
  2005.

\bibitem{cartwright1956structural}
D.~Cartwright and F.~Harary, ``Structural balance: a generalization of heider's
  theory.'' \emph{Psychological review}, vol.~63, no.~5, p. 277, 1956.

\bibitem{MAYBEE198399}
J.~S. Maybee and S.~J. Maybee, ``An algorithm for identifying morishima and
  anti-morishima matrices and balanced digraphs,'' \emph{Mathematical Social
  Sciences}, vol.~6, no.~1, pp. 99--103, 1983.

\bibitem{Farina2011}
L.~Farina and S.~Rinaldi, \emph{Positive Linear Systems: Theory and
  Applications}.\hskip 1em plus 0.5em minus 0.4em\relax John Wiley \& Sons,
  2011, vol.~50.

\end{thebibliography}
\end{document}